\documentclass[a4paper,UKenglish,11pt]{article}
\usepackage{amsfonts}
\usepackage{amssymb}
\usepackage{amsmath}
\usepackage{amsthm}
\usepackage{verbatim}

\theoremstyle{plain}
\newtheorem{theorem}{Theorem}

\newtheorem{proposition}{Proposition}

\usepackage{amsmath,amsfonts,url,amssymb, latexsym}
\usepackage{tikz}
\usepackage[english]{babel}
\usepackage{verbatim}

\usepackage{verbatim}
\usepackage{tikz}

\usepackage{enumerate,graphicx,stfloats,amssymb,graphics,subfig,tikz,xspace,
hyphenat}

\usetikzlibrary{calc}
\usepackage{verbatim}

\usepackage{enumitem}

\newcommand{\dlzerok}{\mathcal{DLR}_{\mathit{reg}}^0[\leq k]}

\newcommand{\dlzerostar}{\mathcal{DLR}_{\mathit{reg}}^0[*]}

\newcommand{\dlzero}{\mathcal{DLR}_{\mathit{reg}}^0}

\newcommand{\UF}{\mbox{$\mbox{\rm U}_1^{}$}}

\usepackage{pgffor}
\usepackage{ifthen}
\usepackage{enumerate}

\newcommand{\UFO}[3]{
\ifthenelse{\equal{#3}{111}}
{$\forall^#1_{TC}[#2]$}{$\forall^#1_{TC}[#2,#3]$}}

\renewcommand{\phi}{\varphi} 

\newcommand{\FOt}{\mbox{$\mbox{\rm FO}^2$}}

\newcommand{\str}[1]{{\mathfrak{#1}}}

\newcommand{\N}{{\mathbb N}}

{\begin{list}{{\rm (\roman{enumiii})}}{\usecounter{enumiii}
}}{\end{list}}

{\begin{list}{{\rm (\alph{enumii})}}{\usecounter{enumii}
}}{\end{list}}

\begin{document}



\title{On the uniform one-dimensional fragment}
%
%
\author{Antti Kuusisto\\
Tampere University and University of Helsinki}
\date{}
%
%
%

\maketitle              
\begin{abstract}
The uniform one-dimensional fragment of first-order logic, $\mathrm{U_1}$, is a
formalism that extends two-variable logic in a natural
way to contexts with relations of all arities. We survey properties of $\mathrm{U_1}$
and investigate its relationship to description logics designed to accommodate
higher arity relations, with particular attention given to  $\mathcal{DLR}_{\mathit{reg}}$.
We also define a description logic version of a variant of $\mathrm{U}_1$ and prove a
range of new results concerning the expressivity of $\mathrm{U}_1$ and
related logics.
%

\end{abstract}

\section{Introduction}

\newcommand{\UFC}{\mbox{$\mbox{\rm UFC}_1^=$}}

%
%
%
%
%
%
%
%

Two-variable logic \cite{IEEEonedimensional:henkin,Sco62}
and the guarded fragment
\cite{IEEEonedimensional:andreka} are
currently perhaps the most widely
studied subsystems of first-order logic.
Two-variable logic $\mathrm{FO}^2$ 
%
%
was proved decidable in \cite{IEEEonedimensional:mortimer},
%
%
%
and the satisfiability problem of 
$\mathrm{FO}^2$ was shown to
be $\mathrm{NEXPTIME}$-complete in \cite{GKV97}.
The extension of two-variable logic with counting quantifiers, $\mathrm{FOC}^2$,
was proved decidable in \cite{IEEEonedimensional:gradel,PST97}
and subsequently shown to be
$\mathrm{NEXPTIME}$-complete in \cite{IEEEonedimensional:pratth}.
Research on extensions and variants of
two-variable logic is currently very active.
Recent research has mainly concerned decidability
and complexity issues
in restriction to particular classes of 
structures and also questions
related to different built-in features and operators that
increase the expressivity of the base language.
Recent articles in the field include for example
\cite{IEEEonedimensional:kieronski,BMSS09,KMP-HT14,IEEEonedimensional:tendera}
%
%
%
%
%
%
and several others.
%

%
%
%

%
The guarded fragment was shown
$\mathrm{2EXPTIME}$-complete in \cite{graadejees}
and in fact $\mathrm{EXPTIME}$-complete over
bounded arity vocabularies in the same article.
The guarded fragment has since then generated a vast literature.
The fragment has recently been significantly
generalized in the article \cite{IEEEonedimensional:barany}
which introduces the
\emph{guarded negation first-order logic} $\mathrm{GNFO}$.
Intuitively, $\mathrm{GNFO}$ only allows
negations of formulae that are guarded
in the sense of the guarded fragment.
The guarded negation fragment
has been shown complete for $\mathrm{2EXPTIME}$
in \cite{IEEEonedimensional:barany}.
The article \cite{kuusistohella2014} introduced the
\emph{uniform one-dimensional fragment}, 
$\UF{}$, which is a natural generalization of $\mathrm{FO}^2$ to 
contexts with relations of arbitrary arities.
Intuitively, $\UF{}$ is a fragment of first-order logic obtained by restricting
quantification to blocks of existential (universal) quantifiers
that \emph{leave
at most one free variable} in the resulting formula.
Additionally, a \emph{uniformity condition} applies to the use of
atomic formulae: if $n,k\geq 2$, then a Boolean combination of
atoms $R(x_1,...,x_k)$ and $S(y_1,...,y_n)$
is allowed only if the sets $\{x_1,...,x_k\}$ and $\{y_1,...,y_n\}$ of
variables are equal.
Boolean combinations of formulae with at most one
free variable can be formed freely,
and the use of equality is unrestricted.
%
%
Several variants of $\mathrm{U}_1$ have also
been investigated in \cite{kuusistohella2014}
and the two subsequent papers \cite{kuusistokieronski2014,kuusistokieronski2015}.
Perhaps the easiest way to  gain
intuitive insight on $\mathrm{U}_1$ is
to consider the \emph{fully uniform fragment}, $\mathrm{FU}_1$,
which is a slight restriction of $\mathrm{U}_1$
introduced in the current article.
It turns out that $\mathrm{FU}_1$ 
can be represented roughly as the
\emph{standard polyadic modal logic} 
where novel accessibility
relations can be formed by the Boolean combination
and permutation of atomic accessibility relations.
Recall that polyadic modal logic is the extension
of modal logic with formulae $\chi:=\langle R\rangle(\varphi_1,...,\varphi_k)$
interpreted such that $M,w\models \chi$
iff there exist points $u_1,...,u_k$ such that 
$(w,u_1,...,u_k)\in R$ and $M,u_i\models\varphi_i$ for each $i$.
It also turns out, as we shall see, that over vocabularies
with at most binary relations, $\mathrm{FU}_1$ is in fact
\emph{equi-expressive} with $\mathrm{FO}^2$.
This result extends a similar observation from
\cite{twovariablemodal} concerning \emph{Boolean
modal logic} with the inverse operator and a
built-in identity modality.
It was proved in \cite{twovariablemodal} that this logic is
expressively complete for $\mathrm{FO}^2$. 
The fact that $\mathrm{FU}_1$ collapses to $\mathrm{FO}^2$
over binary vocabularies can be taken to indicate that $\mathrm{FU}_1$ is a
natural and \emph{in some sense} minimal generalization of $\mathrm{FO}^2$
to higher arity contexts.
The uniform one-dimensional fragment $\mathrm{U}_1$
was shown to have the finite model property and a
$\mathrm{NEXPTIME}$-complete decision problem in
\cite{kuusistokieronski2014}, thereby establishing that
the transition from $\mathrm{FO}^2$ to $\mathrm{U}_1$
comes without a cost in complexity. It was also
shown in \cite{kuusistokieronski2014} that $\mathrm{U}_1$ is
incomparable in expressivity with $\mathrm{FOC}^2$;
we will prove in the current article that $\mathrm{U}_1$ is
incomparable with $\mathrm{GNFO}$, too.
We note, however, that the article \cite{kuusistohella2014} already established a
similar incomparability result concerning $\mathrm{GNFO}$
and the \emph{equality-free} fragment of $\mathrm{U}_1$.
The article \cite{kuusistokieronski2014} 
also showed that the extension of $\mathrm{U}_1$
with counting quantifiers is undecidable. The article \cite{kuusistohella2014}, in
turn, established that relaxing either of the 
two principal constraints of the 
syntax of $\mathrm{U}_1$-formulae---leaving two free
variables after quantification or violating the
uniformity condition---leads to undecidability.
Building on \cite{kuusistohella2014}
and \cite{kuusistokieronski2014}, the article \cite{kuusistokieronski2015} investigated
variants of $\mathrm{U}_1$ in the presence of
built-in equivalence relations. It was shown, e.g.,
that while $\mathrm{U}_1$ 
becomes $\mathrm{2NEXPTIME}$-complete
when a built-in equivalence is added, a certain
natural restriction of $\mathrm{U}_1$ (which still contains $\mathrm{FO}^2$)
remains $\mathrm{NEXPTIME}$-complete.
In the current article we briefly discuss the above
collection of results on $\mathrm{U}_1$ and its variants and
list a number of related open problems.

%

%
%

%
Unlike the guarded fragment and $\mathrm{GNFO}$,
two-variable logic does not cope well
with relations of arities greater than two,
and the same applies to $\mathrm{FOC}^2$.
In database theory
contexts, for example, this can be a major drawback.
Therefore the scope of research on
two-variable logics is significantly restricted.
The uniform one-dimensional fragment $\mathrm{U}_1$
extends two-variable logics in a way that leads to the
possibility of investigating systems with relations of all arities.
%
%

%
Another possible advantage of $\mathrm{U}_1$ is its
one-dimensionality, i.e., the fact that its formulae
are \emph{essentially} of the type $\varphi(x)$, where $x$ is a
free variable. This  links $\mathrm{U}_1$ to
description logics in a natural way, as formulae of $\mathrm{U}_1$
can be regarded as \emph{concepts} in the description logic sense.
Below we make use of this issue and define a
description logic $\mathcal{DL}_{\mathrm{FU}_1}$,
which we prove to be expressively equivalent to the fully
uniform one-dimensional fragment $\mathrm{FU}_1$. 
The logic $\mathcal{DL}_{\mathrm{FU}_1}$ makes
explicit the link between $\mathrm{FU}_1$ and 
polyadic modal logic we mentioned above. It
can be seen as \emph{the canonical extension} of
the description logic $\mathcal{ALBO}^{\mathit{id}}$ \cite{renate}
to higher arity contexts. While $\mathcal{ALBO}^{\mathit{id}}$
is $\mathcal{ALC}$ extended with Boolean and inverse operators on roles,
an identity role and singleton
concepts, $\mathcal{DL}_{\mathrm{FU}_1}$ is essentially the same
system with roles of all arities. The relational inverse operator is generalized to an
operator that slightly generalizes the relational permutation operator.
%
%

%
Higher arity relations arise naturally in contexts
relevant to description logics. Consider for example
the ternary role $R$ such that $R(a,b,c)$ iff $a$ has
contracted a virus $b$ in country $c$, or the quaternary
role $S$ such that $S(c,d,e,f)$ iff $c$ and $d$ have sold $e$ to $f$.
It is easy to see by a counting argument that a $k$-ary relation cannot be encoded by a finite
number of relations of lower arity without changing the domain,
and therefore---in addition to aesthetic considerations---a direct access to higher arity
roles can be advantageous.
Higher arity roles have of course been investigated before in
the description logic literature, for example in
\cite{calvanese,carsten,Schmolze}.
%
%
%
%
Below we compare $\mathcal{DL}_{\mathrm{FU}_1}$
and the system $\mathcal{DLR}_{\mathit{reg}}$
from \cite{calvanese}, which includes, e.g.,
the union, composition and transitive
reflexive closure operators for binary 
roles as well as operators that 
enable the creation of binary relations from
higher arity roles.
%
%
%
We show that $\mathcal{DL}_{\mathrm{FU}_1}$
and $\mathcal{DLR}_{\mathit{reg}}$ are
incomparable in expressivity. While this result itself is
not at all surprizing, it is still worth proving since
the related arguments directly demonstrate
the relative expressivities
of $\mathcal{DLR}_{\mathit{reg}}$ and $\mathcal{DL}_{\mathrm{FU}_1}$.
%
%
%
%
%
%
%
We end the article by identifying a fragment of $\mathcal{DLR}_{\mathit{reg}}$
which is in a certain sense maximal with the property that it
embeds into $\mathcal{DL}_{\mathrm{FU}_1}$. In the context of
this investigation we discuss the curious fact that while $\mathrm{U}_1$
can count, it cannot count well enough to express the
number restriction operators of $\mathcal{DLR}_{\mathit{reg}}$.
In the investigations below concerning expressivity issues, we
make occasional use of the
novel Ehrenfeucht-Fra\"{i}ss\'{e} (EF) game for $\mathrm{U}_1$
from \cite{kuusistokieronski2015}. The related
concrete arguments shed light on the expressivity properties of
%
%
$\mathrm{U}_1$.
%

%
%
%
%
%
Finally, it is worth pointing out here that a rather nice and
potentially fruitful feature of $\mathcal{DL}_{\mathrm{FU}_1}$ is that it is
based on the syntactically and semantically same approach as standard polyadic modal logic.
Thereby $\mathcal{DL}_{\mathrm{FU}_1}$
extends the celebrated and fruitful link between modal and description logics to 
higher arity contexts in a way that \emph{preserves the close relationship} 
between the two fields.

\section{Preliminaries}
We let $\mathrm{VAR}$ denote a countably infinite
set of variable symbols.
Let $X = \{x_1,...,x_k\}$ be a finite set of
variable symbols and let $R$ be an $n$-ary
relation symbol; $R$ is not allowed to be the identity symbol here.
%
%
An atomic formula $R(x_{i_1},...,x_{i_n})$ is called an \emph{$X$-atom}
if $\{x_{i_1},...,x_{i_n}\} = X$.
For example, assuming $x,y,z$ to be
distinct variables, both $S(x,y)$ and $T(x,x,y,y,x)$ are $\{x,y\}$-atoms
while $P(x)$ and $R(x,y,z)$ are not.
Let $\mathbb{Z}_+$ be the set of positive integers.
We let $V$ denote the
infinite relational vocabulary
$V := \bigcup_{k\, \in\, \mathbb{Z}_+} \tau_k$,
where $\tau_k$ is a countably infinite
set of $k$-ary relation symbols; the equality symbol is not in $V$.
A \emph{unary} $V$-atom is an atomic 
formula of the form $P(x)$ or $R(x,...,x)$, where $P,R\in V$.
Here $(x,...,x)$ denotes the tuple that repeats $x$ 
exactly $n$ times, $n$ being the arity of $R$.
%
%
%
%
%
%
%
%

%
The set of formulae of the
\emph{equality-free
uniform one-dimensional fragment} $\mathrm{U}_1({\mathit{wo}=})$
of first-order logic is 
the smallest set $\mathcal{F}$ satisfying the following conditions
(cf. \cite{kuusistohella2014}).
\begin{enumerate}
\item
Every unary $V$-atom is in $\mathcal{F}$.
Also $\bot,\top\in \mathcal{F}$.
%
%
\item
If $\varphi\in \mathcal{F}$, then $\neg\varphi\in\mathcal{F}$.
\item
If $\varphi,\psi\in \mathcal{F}$,
then $(\varphi\wedge\psi)\in \mathcal{F}$.
\item
Let $Y := \{x_0,...,x_k\}\subseteq\mathrm{VAR}$ and $X\subseteq Y$.
Let $\varphi$ be a Boolean combination of {$X$-atoms} over $V$
and {formulae in $\mathcal{F}$ whose free variables (if any) are in $Y$}.
Then
$\exists x_1...\exists x_k\, \varphi\ \in\ \mathcal{F}$ and 
$\exists x_0...\exists x_k\, \varphi\ \in\ \mathcal{F}$.
%
%
%
\end{enumerate}
%
%
%
%
%
%
%
%
For example \[\exists y\exists z((\neg R(x,y,z)
\vee T(z,y,x,x)) \wedge  P(z))\] is a
$\mathrm{U}_1(\mathit{wo}=)$-formula,
while \[\exists y\exists z(S(x,y)\wedge S(y,z)\wedge P(z))\] is not
because $\{x,y\}\not=\{y,z\}$. This latter formula is
said to \emph{violate the uniformity condition} of $\mathrm{U}_1$.
The formula $\exists y R(x,y,z)$ is also illegitimate
because it \emph{violates one-dimensionality},
leaving two variables free instead of one. However,
the sentence $\exists x\exists z\exists y R(x,y,z)$ is legitimate, 
and so is \[\forall x\exists z\exists y(R(x,y,z)\wedge \exists u\neg U(y,u)),\]
while the sentence $\forall x\forall z\exists y R(x,y,z)$ is not.
%
%
%
%
%

%
The \emph{fully uniform one-dimensional fragment} $\mathrm{FU}_1$ is
the logic whose formulae are obtained from
formulae of $\mathrm{U}_1(\mathit{wo}=)$ by allowing the
free substitution of any collection of \emph{binary} relation symbols
by the equality symbol $=$. The
\emph{uniform one-dimensional fragment} $\mathrm{U}_1$ is
obtained by adding to the above four clauses that 
define the set $\mathcal{F}$ of formulae of $\mathrm{U}_1(\mathit{wo} =)$
the additional clause $x=y\ \in\ \mathcal{F}$.
For example \[\exists y\exists z( R(y,z,x) \wedge x\not= y
\wedge \exists z S(y,z))\] is a formula of $\mathrm{U}_1$
but not of $\mathrm{FU}_1$.
Clearly $\mathrm{FU}_1$ is a fragment of $\mathrm{U}_1$.
The following proposition, where $\mathrm{FO}^2$
denotes  two-variable logic with equality, is easy to  prove
using disjunctive normal form representations of formulae.
%
%
%
%
%

%
\begin{proposition}\label{firstproof}
$\mathrm{FU}_1$ and $\mathrm{FO}^2$
are equi-expressive over models with at
most binary relations. That is, in restriction to
models with relations of arity at most two, each
formula of\,  $\mathrm{FU}_1$ with at most
two free variables has an equivalent $\mathrm{FO}^2$-formula,
and each $\mathrm{FO}^2$-formula has an equivalent
$\mathrm{FU}_1$-formula.
\end{proposition}
%
%
%
%
%
%
%
%
However, $\mathrm{U}_1$ is strictly more expressive
than two-variable logic $\mathrm{FO}^2$ even over the empty vocabulary,
because $\mathrm{U}_1$ can count
better than $\mathrm{FO}^2$: we observe that for example
the sentence
\[\exists x\exists y \exists z(x\not=y\wedge x\not= z
\wedge y\not = z)\]
is a $\mathrm{U}_1$-formula.
It is well known and easy to  show by a  two-pebble-game argument
(see \cite{IEEEonedimensional:libkin}  for  pebble games)
that this sentence is not expressible in $\mathrm{FO}^2$.
It is easy to see that $\mathrm{FO}^2$ and
therefore $\mathrm{FU}_1$ can define
the property that $|P| = 1$ for a unary predicate $P$.
Thus \emph{nominals} can be simulated in those logics.
The logic $\mathrm{U}_1$ can 
define even the properties $| P |\leq k$, $| P |\geq k$
and  $| P |= k$ for any finite $k$.
However, the counting capacity of $\mathrm{U}_1$ is
restricted in an interesting way, as we will see later on; $\mathrm{U}_1$
cannot make counting statements about the in-degrees and
out-degrees of binary relations.
Finally, the $\mathrm{U}_1$-sentence \[\exists x\forall y\forall z( R(y,z)
\rightarrow (x=y\vee x=z))\] provides a possibly more
interesting example of what is definable in $\mathrm{U}_1$
but not in $\mathrm{FO}^2$. This sentence states that there is an
element that belongs to every edge of $R$. It is easy to
see by a two-pebble-game argument that this property is
not expressible in $\mathrm{FO}^2$: the \emph{Duplicator} wins the 
two-pebble-game played on $K_2$ and $K_3$, where $K_n$ is
the $n$-clique. Recall that the $n$-clique is
the structure with $n$ elements where $R$ is  the total 
binary relation with the reflexive loops  removed.
\section{Complexity of $\mathrm{U}_1$ and its variants}
The complexity of $\mathrm{U}_1$ 
was identified in \cite{kuusistokieronski2014} by
showing that the logic has the
exponential model property.
\begin{theorem}[\cite{kuusistokieronski2014}]\label{expmodelproperty}
Every satisfiable \UF{}-formula
$\phi$ has a model whose size is bounded exponentially in $|\phi|$. 
\end{theorem}
\begin{theorem}[\cite{kuusistokieronski2014}]\label{complexitytheorem}
The satisfiability problem (=\hspace{0.6mm}finite satisfiability problem) for \UF{} is $\mathrm{NEXPTIME}$-complete.
\end{theorem}
%
%
%
The argument in \cite{kuusistokieronski2014} leading to the above
results bears \emph{at least some degree} of
resemblance to the $\mathrm{NEXPTIME}$
upper bound proof of $\mathrm{FO}^2$ by
Gr\"{a}del, Kolaitis and Vardi in
\cite{GKV97}.
It turns out that $\mathrm{U}_1$-formulae can be transferred
into equisatisfiable formulae in a generalized version of the
\emph{Scott normal form} specially designed for $\mathrm{U}_1$, and
the exponential model property can then be established by
appropriately modifying and extending the arguments applied in 
\cite{GKV97}.
%
%
%
%

%
The complexity results of
the article \cite{kuusistokieronski2014} were 
extended in \cite{kuusistokieronski2015}.
If $L$ denotes a fragment of first-order logic
and $R_1,...,R_k$ are binary relation
symbols, then we let $L(R_1,...,R_k)$ denote the language
obtained by allowing  for the free substitution of identity symbols in
$L$-formulae by the special symbols $R_i$.
%
%
%
The article \cite{kuusistokieronski2015} 
investigated $\mathrm{U}_1$ and  its
variants over models with a built-in equivalence
relation $\sim$. It was shown that
the satisfiability ($\mathrm{SAT}$) and
finite satisfiability ($\mathrm{FINSAT}$) problems for $\mathrm{U}_1(\sim)$
are complete for $\mathrm{2NEXPTIME}$.
The article \cite{kuusistokieronski2015} also
identified a natural restriction $\mathrm{SU}_1$
of $\mathrm{U}_1$ that still extends $\mathrm{FO}^2$
and showed that the $\mathrm{SAT}$
and $\mathrm{FINSAT}$ problems for $\mathrm{SU}_1(\sim)$
are only $\mathrm{NEXPTIME}$-complete;
see \cite{kuusistokieronski2015} for the formal definition of $\mathrm{SU}_1$.
Furthermore, the article \cite{kuusistokieronski2015}
established that the $\mathrm{SAT}$ 
and $\mathrm{FINSAT}$-problems of $\mathrm{SU}_1(\sim_1,\sim_2)$,
i.e., $\mathrm{SU}_1$  
with  two built-in equivalences, is undecidable.
This contrasts with the case for $\mathrm{FO}^2$
which remains decidable
with two equivalences ($\mathrm{SAT}$ 
\cite{kierohullumpi} and $\mathrm{FINSAT}$
\cite{kierohullu}).
Several immediately interesting open problems
remain, for example the decidability issue for $\mathrm{U}_1(\leq)$,
where $\leq$ denotes  a built-in linear order.
Also, while $\mathrm{U}_1(\mathit{tr})$
(i.e., $\mathrm{U}_1$ with a
built-in transitive relation $\mathit{tr}$) was shown
undecidable in \cite{kuusistokieronski2015}, it was left open 
whether $\mathrm{U}_1(\mathit{tr}({\mathit{uniform)})}$ is
decidable; here $\mathrm{U}_1(\mathit{tr}({\mathit{uniform)})}$
denotes the language obtained from $\mathrm{U}_1$ by
allowing the free substitution of any instances of a
binary relation (rather than the equality symbol) by
the built-in transitive relation $\mathit{tr}$.
%

.

\section{Expressivity issues}\label{expressivitysection}
In this section we provide an overview on the
expressivity of $\UF{}$ and its variants.
%
%
The following theorem from \cite{kuusistokieronski2014}
relates the expressivities of $\UF{}$ and $\mathrm{FOC}^2$.
%

%
%

\begin{theorem}[\cite{kuusistokieronski2014}]\label{generalexpressivity}
$\UF{}$ and $\mathrm{FOC}^2$ are incomparable in expressivity.
%
%
\end{theorem}
\begin{proof}
It is easy to show that 
the $\UF{}$-sentence
$\exists x \exists y \exists z R(x,y,z)$
cannot be expressed in $\mathrm{FOC}^2$,
and therefore $\UF{}\not\leq \mathrm{FOC}^2$.
To prove that $\mathrm{FOC}^2\not\leq\UF{}$,
let $S$ be a binary relation symbol.
%
%
We will show that $\UF{}$ cannot express the $\mathrm{FOC}^2$-definable
condition that the in-degree (with respect to the relation $S$) at every node is at most one.  
Assume $\varphi(S)$ is a $\UF{}$-formula that
defines the property. Consider the formula
\[\varphi(S) \wedge \forall x\exists y S(x,y)
\wedge  \exists x \forall y\neg S(y,x).\]
It is clear that this formula has only infinite models,
and thereby the assumption that $\UF{}$ can express $\varphi(S)$ is false
by the finite model property of \UF{} (Theorem \ref{expmodelproperty}).
%
%
\end{proof}
%
%
%

%
We next consider $\UF{}$ over
vocabularies with at most binary relations.
\begin{theorem}[\cite{kuusistokieronski2014}]
Consider models over a relational
vocabulary $\tau$ with
the arity bound two. Suppose that $\tau$ indeed
contains at least one binary relation symbol.
Then $\FOt < \UF < \mathrm{FOC}^2$.
%
%
\end{theorem}
\begin{proof}
We already discussed the strict inclusion $\FOt < \UF{}$ above in
the preliminaries section.
A lengthy proof of the inclusion $\UF \leq \mathrm{FOC}^2$ is
given in \cite{kuusistokieronski2014}. The strictness of this
inclusion follows from
%
%
the proof of Theorem \ref{generalexpressivity} 
where we showed that $\mathrm{U}_1$ cannot
express that the in-degree of a binary relation is at most one.
\end{proof}
We then compare the expressivities of $\UF{}$ and
the guarded negation fragment $\mathrm{GNFO}$
\cite{IEEEonedimensional:barany}.
The first non-inclusion ($\UF{}\not\leq\mathrm{GNFO}$) of
the following theorem has been proved in \cite{kuusistohella2014},
where only the equality-free fragment of $\UF{}$ was investigated.
The second non-inclusion ($\mathrm{GNFO}\not\leq\UF{}$) is new.
\begin{theorem}
$\UF{}$ and $\mathrm{GNFO}$ are incomparable in expressivity.
\end{theorem}
\begin{proof}
Define the two structures
$\bigl(\{a\},\{(a,a)\}\bigr)$ and $\bigl(\{a,b\},\{(a,a),(b,b)\}\bigr)$.
It is straightforward to establish by using the bisimulation for 
$\mathrm{GNFO}$, provided in \cite{IEEEonedimensional:barany}, that
these two structures are bisimilar in the sense of $\mathrm{GNFO}$.
Thus the $\UF{}$-sentence $\exists x\exists y\,\lnot R(x,y)$
is not expressible in $\mathrm{GNFO}$.
Hence $\UF{}\not\leq\mathrm{GNFO}$.
Consider then the $\mathrm{GNFO}$-sentence
\[\varphi := \exists x\exists y \exists z (Rxy\wedge Ryz \wedge Rzx).\]
%
%
%
Let $\mathfrak{A}$ denote the model consisting of
four disjoint copies of the directed cycle with three elements. Let $\mathfrak{B}$
be the model with three disjoint 
copies of the directed cycle with four elements.
It follows rather directly from the Ehrenfeucht-Fra\"{i}ss\'{e}
game for $\UF{}$ (which is defined in \cite{kuusistokieronski2015})
that $\mathfrak{A}$ and $\mathfrak{B}$ satisfy the same $\UF{}$-sentences.
For the game-based argument to work, it is essential that the two models $\mathfrak{A}$
and $\mathfrak{B}$ have the same cardinality, because bijections between subsets
of the domains of $\mathfrak{A}$ and $\mathfrak{B}$ are used in the game.
(See \cite{kuusistokieronski2015} for a detailed 
discussion of the game.)
With $\mathfrak{A}$ and $\mathfrak{B}$ defined in
this way, the rest of the 
game-based argument is straightforward.
We can therefore now conclude that $\UF{}$ cannot
express the $\mathrm{GNFO}$-sentence $\varphi$ we
fixed above, and hence $\mathrm{GNFO}\not\leq\UF{}$.
%
%
%
\end{proof}
Before we close the current section,
we observe that
all the above results concerning expressivity
hold even if attention is limited to finite models only.
The same proofs apply without modification, as the reader can check.
This is especially interesting in
the case of Theorem \ref{generalexpressivity},
whose proof makes use of the finite model property of $\UF{}$.
\section{Undecidability of $\mathrm{U}_1$
with counting quantifiers}\label{undecidabilitysection}


%
%


\newcommand{\LL}{\mathcal{L}}
\newcommand{\G}{\mathfrak{G}}
\newcommand{\NN}{\mathbb{N}}
\newcommand{\cT}{\mathbb{T}}
\newcommand{\TDUF}{\mathcal{UF}_3}
\newcommand{\ODNF}{\mathcal{NF}_1}

Since $\mathrm{FOC}^2$ and $\UF{}$ are both $\mathrm{NEXPTIME}$-complete,
it is natural to ask whether the
extension of \UF{} by counting quantifiers ($\mathrm{UC_1}$) remains decidable.
Formally, $\mathrm{UC_1}$ is obtained from \UF{} by allowing the free substitution of 
quantifiers $\exists$ by quantifiers $\exists^{\geq k},\exists^{\leq k},\exists^{=k}$.
While the transition from $\mathrm{FO}^2$ to $\mathrm{FOC}^2$
preserves $\mathrm{NEXPTIME}$-complete- ness,
the analogous step from $\mathrm{U}_1$ to $\mathrm{UC}_1$
crosses the undecidability barrier.
\begin{theorem}[\cite{kuusistokieronski2014}]\label{theoremundec}
The satisfiability and finite
satisfiability problems of\, $\mathrm{UC}_1$ are $\Pi_1^0$-complete
and $\Sigma_1^0$-complete, respectively.
\end{theorem}
%

%
%
%
%
%
%

%
Thereby $\mathrm{UC}_1$ has the same complexity as
full first-order logic. It is an interesting open problem to
identify natural logics that extend $\mathrm{FOC}^2$
into higher arity contexts in a 
way that preserves decidability. Possible research directions here could
involve for example investigating restrictions of $\mathrm{UC}_1$
based on somewhat more limited ways of using 
the quantifiers $\exists^{\geq k},\exists^{\leq k},\exists^{= k}$.
\section{$\mathrm{U}_1$ and description logics}
In this section we define a
novel logic $\mathcal{DL}_{\mathrm{FU}_1}$ which is a
description logic version of $\mathrm{FU}_1$
and compare it to $\mathcal{DLR}_{\mathit{reg}}$ \cite{calvanese},
which is a well-known description logic that
accommodates higher arity relations.
We first generalize the relational inverse operation to
contexts with higher arity relations.
When $n$ is a
positive integer, we let $[n]$ denote the set $\{1,...,n\}$.
%
%
%
%
We let $\mathrm{SRJ}$
denote the set of all surjections $\sigma:[k]\rightarrow[m]$,
such that $2\leq m\leq k$.
When $m = k$, then $\sigma$ is a permutation;
permutations are natural generalizations of the 
relational inverse operator into higher arity contexts,
and surjections generalize permutations an inch further.
%
%
When we use $\mathrm{SRJ}$ in
constructing the syntax of $\mathcal{DL}_{\mathrm{FU}_1}$ below,
we assume each function $\sigma\in\mathrm{SRJ}$ to be a
suitable string listing the ordered pairs $(n,k)$
such that $\sigma(n) = k$ in binary.
%
%
%
%

%
%
The set $\mathcal{R}$ of roles
of $\mathcal{DL}_{\mathrm{FU}_1}$ is defined by the grammar
$$\mathcal{R}\, ::=\, R\, |\, \varepsilon\, |\, \neg\mathcal{R}\, |\, 
(\mathcal{R}_1\cap\mathcal{R}_2)\, |\, \sigma \mathcal{R}$$
where $R$ denotes an atomic role, $\varepsilon$ the
binary identity role and $\sigma\in\mathrm{SRJ}$.
Here $R$ can have any arity greater or equal to two,
and the arity of $\varepsilon$ is two. The intersection of
relations of different arity will produce the empty relation, so
we may as well allow such terms. (We fix the 
arity of the empty relation in such cases to be two.)
The set of concepts of $\mathcal{DL}_{\mathrm{FU}_1}$ is given by the grammar
$$C\, ::=\, A\, |\, \neg C\, |\, (C_1\sqcap C_2)\, |\, 
\exists \mathcal{R}. (C_1,...,C_n)\, $$
where $A$ is an atomic concept and
the arity of the relation term $\mathcal{R}$ is $n+1$.
An \emph{interpretation} $\mathcal{I}$ is a 
pair $(\Delta,\cdot^{\mathcal{I}})$, where 
$\Delta$ is a nonempty set and $\cdot^{\mathcal{I}}$  a
function such that $R^{\mathcal{I}}\subseteq\Delta^k$
and $A^{\mathcal{I}}\subseteq \Delta$ for atomic
roles $R$ and atomic concepts $A$; here $k$ is the arity of $R$.
The operators of $\mathcal{DL}_{\mathrm{FU}_1}$
are defined as follows.
\begin{enumerate}
\item
$\varepsilon^{\mathcal{I}} := \{\, (u,u)\ |\ u\in\Delta\, \}$,
$(\neg \mathcal{R})^{\mathcal{I}} :=
\Delta^{n+1}\setminus \mathcal{R}^{\mathcal{I}}$
and $(\mathcal{R}_1\cap\mathcal{R}_2)^{\mathcal{I}}
:= \mathcal{R}_1^{\mathcal{I}}\cap \mathcal{R}_2^{\mathcal{I}}$.
%
%
\item
$(\sigma\mathcal{R})^{\mathcal{I}} :=
\{(u_1,...,u_m)\, |\,
(u_{\sigma(1)},...,u_{\sigma(n+1)})\in\mathcal{R}^{\mathcal{I}} \}$.
Here $\sigma$ maps $[n+1]$ onto $[m]$.
The arity of $(\sigma\mathcal{R})^{\mathcal{I}}$ is of course $m$.
%
%
\item
$(\neg C)^{\mathcal{I}} := \Delta\setminus C^{\mathcal{I}}$
and $(C\sqcap D)^{\mathcal{I}} := C^{\mathcal{I}}\cap D^{\mathcal{I}}$.
\item
$(\, \exists \mathcal{R}.(C_1,...,C_{n})\, )^{\mathcal{I}} :=$\\
$\{\, u\in \Delta\, |\, \text{there is a tuple }(u,v_1,...,v_{n})\in \mathcal{R}^{\mathcal{I}}
\text{ s.t. }v_i\in C_i^{\mathcal{I}}\text{ for each }i\, \}$
\end{enumerate}
In the pathological 
case where $\sigma:[n]\rightarrow[m]$ acts on a relation $\mathcal{R}$
whose arity is not equal to $n$, the empty binary relation is produced. 
We need the surjection operators (rather than simply permutations) in
order to express in $\mathcal{DL}_{\mathrm{FU}_1}$
conditions such as the one
given by the $\mathrm{FU}_1$-formula $\exists y
(R(x,y)\wedge S(x,y,x) \wedge P(y))$.
In the following theorem, equivalence means
equivalence in the standard sense in which formulae of modal and
predicate logic are compared.
\begin{theorem}\label{timegoesbytheorem}
$\mathcal{DL}_{\mathrm{FU}_1}$ and $\mathrm{FU}_1$
are equi-expressive: each $\mathrm{FU}_1$-formula $\varphi(x)$
has an equivalent $\mathcal{DL}_{\mathrm{FU}_1}$-concept,
and vice versa.
\end{theorem}
\begin{proof}
We only provide a rough sketch the proof.
The most involved 
issue here is the translation of
$\mathrm{FU}_1$-formulae of
the type $\exists x_1...\exists x_k \varphi$
into $\mathcal{DL}_{\mathrm{FU}_1}$,
where $\varphi$ is a Boolean combination
of higher arity atoms and 
at most unary $\mathrm{FU}_1$-formulae. 
Here we put $\varphi$ into disjunctive normal form
and distribute the quantifier prefix over the
disjunctions in order to obtain a disjunction of
formulae of the type
$$\exists x_1\dots \exists x_k
(\mathcal{T}(y_1,\dots ,y_n)\wedge \chi_1(u_1)\wedge
\dots \wedge\chi_m(u_m)\bigr)$$
such that the following three conditions hold.
\begin{itemize}
\item
$\{y_1,\dots ,y_n\}\subseteq\{x_0,x_1,\dots ,x_k\}$
\item
$\{u_1,\dots ,u_m\}\subseteq\{x_0,x_1,\dots ,x_k\}$ 
\item 
The term $\mathcal{T}(y_1,\dots ,y_n)$ is a
\emph{conjunction} of higher arity literals
(atoms and negated atoms) such that each literal
has exactly the same set $\{y_1,...,y_n\}$ of variables.
\end{itemize}
Such formulae 
can easily be translated into $\mathcal{DL}_{\mathrm{FU}_1}$,
assuming inductively that we already know how to translate the 
unary $\mathrm{FU}_1$-formulae $\chi_i(u_i)$.
\end{proof}
We then define the description logic $\mathcal{DLR}_{\mathit{reg}}$
from \cite{calvanese}
and compare it to $\mathcal{DL}_{\mathrm{FU}_1}$.
$\mathcal{DLR}_{\mathit{reg}}$ is defined  by the grammar
\begin{align*}
\mathcal{R}\ &::=\ \top_n\ |\ R\ |\ (\$i/n:C)\ |\ \neg\mathcal{R}\ |\
(\mathcal{R}_1\cap\mathcal{R}_2)\\
\mathcal{E}\ &::=\ \varepsilon\ |\
\mathcal{R}_{|\$i,\$j}\ |\ (\mathcal{E}_1\circ\mathcal{E}_2)\ |
\ (\mathcal{E}_1\cup\mathcal{E}_2)\ |\ \mathcal{E}^*\\
C\ &::=\ \top_1\ |\ A\ |\ \neg C\ |\ (C_1\sqcap C_2)
\ |\ \exists\mathcal{E}.C\ |\  \exists[\$i]\mathcal{R}\ |
\ (\leq k\, [\$i]\mathcal{R})
\end{align*}
where $R$ is an atomic role and $A$
an atomic concept from a finite set $\mathcal{V}$ of
atomic role and concept symbols.
The indices $i$ and $j$ denote 
integers between $1$ and $n_{\mathit{max}}$ 
(where $n_{\mathit{max}}$ is the maximum  arity
of  the symbols in $\mathcal{V}$), $n$
denotes an integer between $2$ and $n_{\mathit{max}}$
and $k$ denotes a non-negative integer.
All these numbers are encoded in binary.
An interpretation $\mathcal{I} = (\Delta,\cdot^{\mathcal{I}})$
for $\mathcal{DLR}_{\mathit{reg}}$ over $\mathcal{V}$ is 
any structure such that the following conditions are met (cf. \cite{calvanese}).
\begin{enumerate}
\item
For each atomic concept $A\in\mathcal{V}$ and atomic role $R\in\mathcal{V}$,
we have $A\subseteq\Delta$ and $R\subseteq\Delta^n$, where $n$ is
the arity of $R$.
\item
For each $n>1$, $(\top_n)^{\mathcal{I}}$
is  a subset of $\Delta^n$
that covers the relations of arity $n$.
%
%
%
%
\item
$(\$i/n:C)^{\mathcal{I}}$ is the set of tuples $(u_1,...,u_n)\in(\top_n)^{\mathcal{I}}$
such that $u_i\in C^{\mathcal{I}}$.
\item
$(\neg\mathcal{R})^{\mathcal{I}} =
(\top_n)^{\mathcal{I}}\setminus\mathcal{R}^{\mathcal{I}}$
when $\mathcal{R}$ is an $n$-ary term
and $(\mathcal{R}_1\cap\mathcal{R}_2)^{\mathcal{I}}
= \mathcal{R}_1^{\mathcal{I}}\cap\mathcal{R}_2^{\mathcal{I}}$.
\item
$\varepsilon^{\mathcal{I}} = \{\, (u,u)\, |\, u\in\Delta\, \}$ and
$(\mathcal{R}_{|\$i,\$j})^{\mathcal{I}}$ is the relation
%

%
\begin{center}
$\{\, (u,v)\, |\, 
u =w_i\text{ and }v=w_j\text{ for some tuple }(w_1,...,w_n)
\in\mathcal{R}^{\mathcal{I}}\, \}.$
\end{center}
%


%
\item
The operators $\circ$, $\cup$ and $\cdot^*$ in the
terms $(\mathcal{E}_1\circ\mathcal{E}_2)$, $(\mathcal{E}_1\cup\mathcal{E}_2)$
and $\mathcal{E}^{*}$
are interpreted in the usual way, i.e., $\circ$ is the 
relational composition operator, $\cup$ 
the union and $\cdot^*$ the
transivitive reflexive closure operator.
\item
$(\top_1)^{\mathcal{I}} =  \Delta$,
$(\neg C)^{\mathcal{I}} = (\top_1)^{\mathcal{I}}\setminus C^{\mathcal{I}}$
and $(C\sqcap D)^{\mathcal{I}} = C^{\mathcal{I}}\cap D^{\mathcal{I}}$.
\item
$(\exists\mathcal{E}.C)^{\mathcal{I}}
= \{\, u\, |\, \text{ exists }(u,v)\in\mathcal{E}^{\mathcal{I}}
\text{ such that }v\in C^{\mathcal{I}}\, \}$
\item
$(\exists[\$i]\mathcal{R})^{\mathcal{I}} = \{\, u\, |\, \text{ exists }
(v_1,...,v_n)\in\mathcal{R}^{\mathcal{I}}
\text{ such that }u = v_i\}$
\item
$(\leq k\, [\$i]\mathcal{R})^{\mathcal{I}} = \{\, u\, |\, 
\,  \, \, |\{\, u\, |\, \text{exists }(v_1,...,v_n)\in\mathcal{R}^{\mathcal{I}}
\text{ s.t. }u = v_i\}|\, \leq k\, \}$.
\end{enumerate}
%

%
%
%
$\mathcal{DLR}_{\mathit{reg}}$ interpretations
are associated with the atomic built-in relations $\top_n$.
%
%
When comparing the expressivity of $\mathcal{DLR}_{\mathit{reg}}$
with $\mathcal{DL}_{\mathrm{FU}_1}$ below, we
consider interpretations $\mathcal{I}$ where
the relations $\top_n$ are appropriate atomic built-in roles 
and thus directly available
also in $\mathcal{DL}_{\mathrm{FU}_1}$.
%

%
%

%
\begin{proposition}\label{propositionsomething}
$\mathcal{DLR}_{\mathit{reg}}$ and $\mathcal{DL}_{\mathrm{FU_1}}$ are
incomparable in expressvity.
\end{proposition}
\begin{proof}
It is easy to see that $\mathcal{DLR}_{\mathit{reg}}$ is
closed under disjoint copies such that if $C^{\mathcal{I}} = U$
for some $\mathcal{DLR}_{\mathit{reg}}$-concept $C$,
then $C^{{\mathcal{I}_1+\mathcal{I}_2}} = U_1\cup U_2$,
where $\mathcal{I}_1+\mathcal{I}_2$
consists of two disjoint copies of $\mathcal{I}$ and 
obviously $U_1$ and $U_2$
are the related copies of $U$.
Because of the free use of role negation in $\mathcal{DL}_{\mathrm{FU_1}}$,
the same does not hold in that logic. For example
the $\mathcal{DL}_{\mathrm{FU}_1}$-concept $\neg\exists(\neg R).A$,
where $R$ is a binary role, is
satisfied in an interpretation consisting of a
single element $u$ that
satisfies $A$ and connects to itself via $R$. This interpretation
together with a disjoint copy of itself does not satisfy $\neg\exists(\neg R).A$.
Thus $\mathcal{DL}_{\mathrm{FU_1}}$ is not contained in $\mathcal{DLR}_{\mathit{reg}}$.
For the converse, it suffices to observe that $\mathcal{DL}_{\mathrm{FU_1}}$
cannot define the concept $\exists(R^*).A$. It is well known
that this property is not first-order expressible, and thus it is not
definable in $\mathcal{DL}_{\mathrm{FU_1}}$.
\end{proof}
We finish up the current section by identifying a
maximal fragment of $\mathcal{DLR}_{\mathit{reg}}$ that
embeds into $\mathcal{DL}_{\mathrm{FU}_1}$.
What exactly we mean
by maximality  in this context
will become clear below. 
Let $\dlzero$ denote the
fragment of $\mathcal{DLR}_{\mathit{reg}}$ without 
Kleene star and counting, i.e., $\dlzero$ is
obtained by the grammar that drops the terms $\mathcal{E}^*$
and $(\leq k\, [\$i]\mathcal{R})$
from the grammar of $\mathcal{DLR}_{\mathit{reg}}$.
For each \emph{positive} integer $k$, we
let $\dlzerok$ denote the system we obtain if we
add the terms $(\leq k\, [\$i]\mathcal{R})$
(with each arity for $\mathcal{R}$
and each related $i$ included) to  the grammar of $\mathcal{DLR}_{\mathit{reg}}^0$.
(Note that $(\leq 0\, [\$i]\mathcal{R})$
is equivalent to $\neg \exists[\$i]\mathcal{R}$.)
Similarly, we let $\dlzerostar$ be the logic we obtain  by
adding the term $\mathcal{E}^*$ to the
grammar of $\mathcal{DLR}_{\mathit{reg}}^0$.
%

%
%
%
%

%
We will show that while $\dlzero$
embeds into $\mathcal{DL}_{\mathrm{FU}_1}$ (Theorem \ref{secondlast}),
neither $\dlzerostar$ 
nor any of the logics $\dlzerok$ does (Theorem \ref{last}).
%
%
%
%
%
%
%
%
%
%
We already observed above that 
the operator $\cdot^*$ of $\mathcal{DLR}_{\mathit{reg}}$ is
inexpressible in $\mathcal{DL}_{\mathrm{FU_1}}$.
The fact that the number restriction operators $(\leq k\, [\$i]\mathcal{R})$ are
definable neither in $\mathcal{DL}_{\mathrm{FU}_1}$
nor in $\mathrm{U}_1$, as we shall prove, is somewhat
more surprising since $\mathrm{U}_1$
can do \emph{some} counting. However, as we already 
%
%
discussed earlier, the counting ability of $\mathrm{U}_1$ is limited.
%
%

%
Finally, it is not entirely trivial
that we can indeed keep the composition
operator in $\dlzero$
and still embed this logic into $\mathcal{DL}_{\mathrm{FU}_1}$.
This is because the use of the composition operator often
requires the three-variable fragment of first-order logic,
and $\mathcal{DL}_{\mathrm{FU}_1}$ collapses to $\mathrm{FO}^2$ on
binary vocabularies.
\begin{theorem}\label{secondlast}
$\dlzero$ 
embeds into $\mathcal{DL}_{\mathrm{FU}_1}$.
\end{theorem}
\begin{proof}
We begin by showing 
that we can eliminate the composition operator $\circ$
from $\dlzero$ altogether. Consider a
concept $D$ of $\mathcal{DLR}_{\mathit{reg}}^0$.
We first observe that we can use the
the standard identity
$R\circ (S\cup T) = (R\circ S)\cup (R\circ T)$
of relation algebra to obtain from $D$ an expression 
where the composition operators are on the ``atomic" level,
with the relational terms $\varepsilon$ and $\mathcal{R}_{|\$i,\$j}$ of
the grammar of $\mathcal{DLR}_{\mathit{reg}}$
regarded as atoms.
%
%
%
%
%
%
%
We then use the
equivalence
\[\exists\bigl(\mathcal{E}_1
\cup...\cup\mathcal{E}_m\bigr).\, C\, \equiv\, 
(\exists\mathcal{E}_1.C)\ \sqcup...\sqcup (\exists\mathcal{E}_m.C)\]
to obtain a disjunction of formulae $\exists\mathcal{E}_i. C$
where $\mathcal{E}_i$ is a
composition of ``atomic" terms $\mathcal{S}$.
%
%
To eliminate the composition operators from
the terms $\mathcal{E}_i = \mathcal{S}_1\circ...\circ\mathcal{S}_n$,
we use the equivalence
\[\exists\bigl(\, \mathcal{S}_1\circ...
\circ\mathcal{S}_n\, ).C\, \equiv\, 
\exists \mathcal{S}_1.\exists \mathcal{S}_2.\exists\mathcal{S}_3\ \ ...\ \
\exists\mathcal{S}_n. C.\]
%
%
%
%
%
%
%
Thus we can eliminate instances of $\circ$
from $\mathcal{DLR}_{\mathit{reg}}^0$.
%
%
%

%
Next we note that all the remaining union operators are
also eliminable, using the 
equivalence \[\exists\bigl(\mathcal{E}_1
\cup...\cup\mathcal{E}_m\bigr).\, C\, \equiv\, 
(\exists\mathcal{E}_1.C)\ \sqcup...\sqcup (\exists\mathcal{E}_m.C).\]
We then show how to translate the
obtained formula (which is free of
union and composition operators) into $\mathcal{DL}_{\mathrm{FU}_1}$.
For presentational reasons, we
will translate the formula
into the first-order fragment $\mathrm{FU}_1$.
The syntax of $\mathcal{DLR}_{\mathit{reg}}^0$
without composition and union is
given by the grammar
\begin{align*}
\mathcal{R}\ &::=\ \top_n\ |\ R\ |\ (\$i/n:C)\ |\ \neg\mathcal{R}\ |\
(\mathcal{R}_1\cap\mathcal{R}_2)\\
\mathcal{E}\ &::=\ \varepsilon\ |\
\mathcal{R}_{|\$i,\$j}\\
%
%
C\ &::=\ \top_1\ |\ A\ |\ \neg C\ |\ (C_1\sqcap C_2)
\ |\ \exists\mathcal{E}.C\ |\  \exists[\$i]\mathcal{R}\ 
\end{align*}
where $\mathcal{R}_{|\$i,\$j}$ with $i=j$ is not 
allowed; these are easy to eliminate. Our translation will be defined
with three translation operators $s,t,T$ that correspond to,
respectively,
the terms for $\mathcal{R},\mathcal{E},C$ above.
Each of these operators is parameterized by an
appropriate tuple of variables.
We first define $T$ as follows.
\begin{enumerate}
\item
$T[x](\top_1) := \top$ and $T[x](A) := A(x)$.
\item
$T[x](\neg C) := \neg T[x](C)$
and $T[x]( C_1\sqcap C_2 ) := T[x](C_1)\wedge T[x](C_2)$.
\item
$T[x](\exists\mathcal{E}.C) :=
\exists y\bigl(
t[x,y](\, \mathcal{E}\, )\wedge T[y]C\, \bigr)$, where $t$ is
the translation for terms $\mathcal{E}$ to be defined below.
\item
\begin{align*}
T[x](\exists[\$i]\mathcal{R}) & \\  := & 
\exists x_1...\exists x_{i-1}\exists x_{i+1}...\exists x_n\bigl(\,
s[x_1,...,x_{i-1},x,x_{i+1},...,x_n](\mathcal{R})\, \bigr),\end{align*}
where $s$ is a translation for $\mathcal{R}$
and $n$ is the arity of $\mathcal{R}$.
\end{enumerate}
We then define the operator $t$.
%
%
%
%
%
%
%
%
%
%
\begin{enumerate}
\item
$t[x,y](\varepsilon) := x = y$.
\item
$t[x,y](\mathcal{R}_{|\$i,\$j}) :=
\exists \overline{z}\bigl( s[\overline{u}](\mathcal{R}))$,
where $\exists\overline{z}$ quantifies existentially
each of the variables $x_1,...,x_n$ except for $x_i$ and $x_j$,
and where $\overline{u}$ is obtained from the tuple
$(x_1,...,x_n)$ by replacing $x_i$ by $x$ and $x_j$ by $y$.
Here $n$ is the arity of the relation $\mathcal{R}$
and $s$ is the translation for $\mathcal{R}$.
%
%
%
%
%
\end{enumerate}
We finally define the operator $s$ as follows.
\begin{enumerate}
\item
$s[x_1,...,x_n](\top_n) := \top_n(x_1,...,x_n)$
and $s[x_1,...,x_n](R) := R(x_1,...,x_n)$ for
atomic roles $R$ and the built-in relation $\top_{n}$.
\item
$s[x_1,...,x_n](\, (\$i/n:C)\, ) := T[x_i](C)
\wedge\top_n(x_1,...,x_n)$, where $T$ is
the translation for $C$.
\item
$s[x_1,...,x_n](\neg\mathcal{R}) :=
\top_n(x_1,...,x_n)\wedge \neg\, s[x_1,...,x_n](\mathcal{R})$.
\item
$s[x_1,...,x_n](\, \mathcal{R}_1\cap\mathcal{R}_2\, )
:= s[x_1,...,x_n](\mathcal{R}_1)\wedge s[x_1,...,x_n](\mathcal{R}_2)$.
\end{enumerate}
The translated formula is now easily modified to a
formula of $\mathrm{FU}_1$.
This involves shifting the quantifiers introduced in clause 2 of
the translation  $t[x,y]$.
%
%
%
%
%
%
%
\end{proof}
We then show that 
none of the operators of $\mathcal{DLR}_{\mathit{reg}}$
missing from $\dlzero$ could be added to $\dlzero$
without losing the embedding into $\mathcal{DL}_{\mathrm{FU}_1}$.
By an operator we here mean $\cdot^*$ 
and each term $(\leq k\, [\$i]\mathcal{R})$
with $k\in\mathbb{Z}_+$. Note that 
for a fixed $k$, the term $(\leq k\, [\$i]\mathcal{R})$
strictly speaking denotes a collection of
operators, because we could vary $i$ and the arity of $\mathcal{R}$.
Thus a more fine-grained analysis than the one below
could be given. We ignore this issue
for the sake of simplicity.
%

%
%
\begin{theorem}\label{last}
$\dlzerostar$
and $\dlzerok$ for each $k\in\mathbb{Z}_+$
are all incomparable with $\mathcal{DL}_{\mathrm{FU}_1}$
\end{theorem}
%
%
%
%
\begin{proof}
We already observed in the proof of
Proposition \ref{propositionsomething} 
that $\mathcal{DL}_{\mathrm{FU}_1}$ cannot
define the concept $\exists(R^*).A$
and that $\mathcal{DLR}_{\mathit{reg}}$ cannot
define $\neg\exists(\neg R).A$,
where $\neg$ is the full negation of $\mathcal{DL}_{\mathrm{FU}_1}$.
Thus it now suffices to show that for each $k\in\mathbb{Z}_+$,
the concept $(\leq k\, [\$2] R)$ is
not expressible in $\mathcal{DL}_{\mathrm{FU}_1}$.
Here $R$ is a binary relation.
In the proof of Theorem \ref{generalexpressivity},
we already dealt with the special case where $k = 1$:
if $\varphi(x)$ was an $\mathrm{FU}_1$-formula defining the 
concept $(\leq 1\, [\$2] R)$, then the $\mathrm{FU}_1$-sentence
$\forall x\varphi(x)$ would define that the in-degree of $R$ is at most one.
Thus we can now fix a $k\geq 2$
and define two interpretations, one consisting of $k+1$
copies of the clique of size $k$ and the other one of $k$
copies of the clique of size $k+1$.
(Recall that a clique is a structure where the
binary relation $R$ is the total
relation with the reflexive loops removed).
We have prepared the setting in such a way that it is now
easy to show, using once again the
%
%
EF-game for $\mathrm{U}_1$ (defined in \cite{kuusistokieronski2015}),
that the two structures satisfy exactly the same $\mathrm{U}_1$-sentences.
%
%
However, the concept $(\leq k-1\, [\$2]R)$ is satisfied by 
every element in the first structure and by none of the elements of 
the second one. Thus no $\mathrm{U}_1$-formula is
equivalent to $(\leq k-1\, [\$2]R)$, because if $\varphi(x)$
was equivalent to $(\leq k-1\, [\$2]R)$,
the $\mathrm{U}_1$-sentence $\exists x\varphi(x)$
would be satisfied by the first structure but not the  second one.
\end{proof}

Logics building on the two-variable logic, $\mathrm{U}_1$ and
$\mathcal{DL}_{\mathrm{FU}_1}$ soon become \emph{highly} undecidable if
extended with fixed points or even transitive closure. 
Notable exceptions to this pattern can be obtained via
the game-theoretic recursion (see, e.g., the article
\cite{compl} for the definitions, and note that here we are 
mainly---but not exclusively---interested in the 
unbounded semantics). The related logics have many nice properties.
For example validity for the so 
extended two-variable logic is
complete for co-NEXPTIME \cite{compl}, and we get recursive enumerability of
validities even for the extension of FO with game-theoretic recursion.

The logics with game-theoretic recursion constitute an  
interesting collection of systems with self-referential statements. 
Note that the formula $C\neg C$ gives a formulation of the liar paradox. 
Within the standard semantics (cf. \cite{compl}), this is an
indeterminate statement. To consider the paradox generally, let LS stand for the informal
claim ``this sentence is not true''. Note that LS can be interpreted as
the claim ``this sentence is not $T_1$'' where $T_1$ refers to
some reasonably ordinary notion of truth,
let us call it ``first-level truth''. This could be, for
example, the ``well-founded truth'' \cite{rules}, but various notions---formal and
informal---are fine here.  With this interpretation, we can 
associate LS with a second-level truth value $U_2$. Here $U_2$ can mean, for 
example, ``second-level true'' or ``second-level false'', or
perhaps ``not second-level true'' or ``not second-level false'', to give a 
few options out of many. In each case, we
avoid the usual paradox associated with the 
claim ``this sentence is not true''. This is because $U_2$ 
operates on a different level than $T_1$, that is, gives a 
\emph{different sense} of truth. 
So ``second-level true''
means true on some other level (or in a
different sense) than ``first-level true'' $T_1$, and similarly for falsity and
other truth notions.   
Altogether, the above discussion 
gives one suggestion for a resolution of the paradox
associated with ``this sentence is not true''. In general,
there is no paradox in stating, to give one example, that it is not
second-level true that ``this sentence is not first-level true''.
It is not so central what second-level truth value we give 
``this sentence is not true'', the point is that we can give it
any second-level truth value, as it does not interact with a first-level
truth values. 
The same analysis gives---for similar reasons---a similar
suggestion for resolving the paradox with ``this sentence is false''
(where we talk about falsity instead of not being true).

Let us get back to the game-theoretic recursion and give a 
related example. 
Considering the standard semantics for $C\neg C$, we can take
``first-level true'' be the well-founded truth 
mentioned above, meaning Eloise having a winning strategy.
We can also take ``first-level false'' to
mean a similar, well-founded notion meaning that 
Abelard has a winning strategy. 
Since $C\neg C$ is indeterminate,
the sentence $C\neg C$ is neither ``first-level true''
nor ``first-level false''. Thus it is 
not ``first-level true'', and it is not ``first-level false''.
However, to give one possibility, $C\neg C$ can,
without problems, be considered ``second-level false'', 
with intuition would be that Eloise does not have a winning strategy, i.e.,
``second-level false'' is equated with not being ``first-level true''.
On the other hand, giving a somewhat different interpretation to second-level
truth values, $C\neg C$ can also be considered to 
\emph{not} be ``second-level false''
since Abelard has no winning strategy (so here we 
equate ``second-level false'' with 
``first-level false'' on the formal level). And so on. The
general point is simply that we
resolve the paradox by suggesting to give new truth values that 
operate on different levels (or in different senses). This is often
natural, as for example ``well-founded truth'' (a possible interpretation for ``first-level true'') and similarly ``well-founded falsity''
are naturally 
different from many second-level truth values that talk about 
indeterminacy due to neither player having a winning strategy. 
And, in the general case, since the truth values on different levels do not interact (in any forced way), 
there is no paradox. There is indeed no problem stating (for example) that the 
sentence ``this sentence is not first-level true'' is 
not second-level true. Indeed, we can read $C\neg C$ as stating
that ``this sentence is not first-level true'' and then give it
the semantic judgement that it is not ``second-level true'', where ``second-level
true'' means Eloise having a winning strategy. And so on, with many 
interesting different options. 
The general resolution simply says that we can give a second-level truth value in more or less any way we like. Some options can be more natural than others, but there is no
inconsistency. Indeed, defining what positive real numbers are extends 
the definition of positive integers, and that generalization is natural but not
something that the properties of integers force. We could extend the
definition in other, alternative ways without direct inconsistency.

There are many suggested resolutions of the general paradox in
the literature. In particular, we note
that there is a significant difference between Tarski's
suggested resolution and the suggestion discussed above. In Tarski's 
suggestion, the liar sentence is denied the
role of an admissible sentence, as the
truth predicate should not apply on any level $n$ to sentences on the same level $n$. 
This kind of a constraint is not used in the suggestion we
discussed above. Note also
that our suggestion indeed bears a nice link of the syntax and 
semantics of the $C$-operator as used in \cite{compl} and the 
articles leading to that paper (see the introduction of
\cite{compl} for more information). 
In general, the framework with $C$ relates nicely to many issues on
self-reference. However, it of course does not touch all issues. For example, in
the related logics (so far), $\varphi$ is equivalent to $\varphi\wedge\varphi$.
This breaks with sentences such as ``this sentence has less
than fifty letters'' which is true but its conjunction with itself is---at least under one natural intepretation---not true.

Concerning the syntax of logics with the game-theoretic looping operator $C$, it is
sometimes convenient to write $L_C$ for the non-atomic occurrences of $C$
while atomic operators are kept as they are. Now $C\neg C$ becomes $L_C \neg C$.
This formula fits both the predicate logic as well as the 
description logic context. Concerning the latter 
family of logics, for
example the formula $C \exists R.(A \sqcup C)$ becomes $L_C \exists R. (A \sqcup C)$.
Eloise wins the game for this formula precisely from those points
from where we can reach a point satisfying the atomic
concept $A$. Note that $C$ is here of
course an iterator variable, not a metasymbol denoting a concept. To 
avoid confusion, it may indeed be convenient to use some other letter 
for iterator variables than $C$ in the context of description logics.

\bibliography{uodf}

\begin{thebibliography}{10}

\bibitem{IEEEonedimensional:andreka}
Hajnal Andr\'eka, Johan van Benthem, and Istvan N\'{e}meti.
\newblock Modal languages and bounded fragments of predicate logic.
\newblock {\em Journal of Philosophical Logic}, 27(3):217--274, 1998.

\bibitem{IEEEonedimensional:barany}
Vince B{\'a}r{\'a}ny, Balder ten Cate, and Luc Segoufin.
\newblock Guarded negation.
\newblock In {\em Proc. of ICALP (2)}, pages 356--367, 2011.

\bibitem{IEEEonedimensional:kieronski}
Saguy Benaim, Michael Benedikt, Witold Charatonik, Emanuel Kiero\'{n}ski,
  Rastislav Lenhardt, Filip Mazowiecki, and James Worrell.
\newblock Complexity of two-variable logic on finite trees.
\newblock In {\em Proc. of ICALP (2)}, pages 74--88, 2013.

\bibitem{BMSS09}
M.~Boja\'{n}czyk, A.~Muscholl, T.~Schwentick, and L.~Segoufin.
\newblock Two-variable logic on data trees and {XML} reasoning.
\newblock {\em Journal of the ACM}, 56(3), 2009.

\bibitem{calvanese}
Diego Calvanese, Giuseppe {De Giacomo}, and Maurizio Lenzerini.
\newblock On the decidability of query containment under constraints.
\newblock In {\em Proc. of PODS}, pages 149--158, 1998.

\bibitem{GKV97}
E.~Gr{\"a}del, P.~Kolaitis, and M.~Vardi.
\newblock On the decision problem for two-variable first-order logic.
\newblock {\em Bulletin of Symbolic Logic}, 3(1):53--69, 1997.

\bibitem{graadejees}
Erich Gr{\"{a}}del.
\newblock On the restraining power of guards.
\newblock {\em Journal of Symbolic Logic}, 64(4):1719--1742, 1999.

\bibitem{IEEEonedimensional:gradel}
Erich Gr{\"a}del, Martin Otto, and Eric Rosen.
\newblock Two-variable logic with counting is decidable.
\newblock In {\em Proc. of LICS}, pages 306--317, 1997.

\bibitem{kuusistohella2014}
Lauri Hella and Antti Kuusisto.
\newblock One-dimensional fragment of first-order logic.
\newblock In {\em Proc. of AiML}, pages 274--293, 2014.

\bibitem{IEEEonedimensional:henkin}
Leon Henkin.
\newblock Logical systems containing only a finite number of symbols.
\newblock {\em Presses De l'Universit\'{e} De Montr\'{e}al}, 1967. 

\bibitem{compl}
Reijo Jaakkola and Antti Kuusisto.
\newblock First-order logic with self-reference.
\newblock arXiv:2207.07397, 2022.

\bibitem{KMP-HT14}
E.~Kiero\'{n}ski, J.~Michaliszyn, I.~Pratt-Hartmann, and L.~Tendera.
\newblock Two-variable first-order logic with equivalence closure.
\newblock {\em SIAM Journal on Computing}, 43(3), 2014.

\bibitem{kuusistokieronski2014}
Emanuel Kiero\'{n}ski and Antti Kuusisto.
\newblock Complexity and expressivity of uniform one-dimensional fragment with
  equality.
\newblock In {\em Proc. of MFCS (1)}, pages 365--376, 2014.

\bibitem{kuusistokieronski2015}
Emanuel Kiero\'{n}ski and Antti Kuusisto.
\newblock Uniform one-dimensional fragments with one equivalence relation.
\newblock In {\em Proc. of CSL}, pages 597--615, 2015.

\bibitem{kierohullumpi}
Emanuel Kiero\'{n}ski and Martin Otto.
\newblock Small substructures and decidability issues for first-order logic
  with two variables.
\newblock {\em Journal of Symbolic Logic}, 77(3):729--765, 2012.

\bibitem{kierohullu}
Emanuel Kiero\'{n}ski and Lidia Tendera.
\newblock On finite satisfiability of two-variable first-order logic with
  equivalence relations.
\newblock In {\em Proc. of LICS}, pages 123--132, 2009. 

\bibitem{rules}
Antti Kuusisto.
\newblock The power of clockings. 
\newblock arXiv:2201.08353v2, 2023. 

\bibitem{IEEEonedimensional:libkin}
Leonid Libkin.
\newblock {\em Elements of Finite Model Theory}.
\newblock Springer, 2004.

\bibitem{carsten}
Carsten Lutz, Ulrike Sattler, and Stephan Tobies.
\newblock A suggestion for an n-ary description logic.
\newblock In {\em Proc. of DL'99}, 1999.

\bibitem{twovariablemodal}
Carsten Lutz, Ulrike Sattler, and Frank Wolter.
\newblock Modal logic and the two-variable fragment.
\newblock In {\em Proc. of {CSL}}, pages 247--261, 2001.

\bibitem{IEEEonedimensional:mortimer}
M.~Mortimer.
\newblock On languages with two variables.
\newblock {\em Zeitschrift f\"{u}r Mathematische Logik und Grundlagen der
  Mathematik}, 21(1):135--140, 1975.

\bibitem{PST97}
L.~Pacholski, W.~Szwast, and L.~Tendera.
\newblock Complexity of two-variable logic with counting.
\newblock In {\em Proc. of LICS}, pages 318--327, 1997.

\bibitem{IEEEonedimensional:pratth}
Ian Pratt-Hartmann.
\newblock Complexity of the two-variable fragment with counting quantifiers.
\newblock {\em Journal of Logic, Language and Information}, 14(3):369--395,
  2005.

\bibitem{renate}
Renate~A. Schmidt and Dmitry Tishkovsky.
\newblock Using tableau to decide description logics with full role negation
  and identity.
\newblock {\em {ACM} Transactions on Computational Logic}, 15(1), 2014.

\bibitem{Schmolze}
James~G. Schmolze.
\newblock Terminological knowledge representation systems supporting n-ary
  terms.
\newblock In {\em Proc. of KR'89}, pages 432--443, 1989.

\bibitem{Sco62}
D.~Scott.
\newblock A decision method for validity of sentences in two variables.
\newblock {\em Journal of Symbolic Logic}, 27, 1962.

\bibitem{IEEEonedimensional:tendera}
Wieslaw Szwast and Lidia Tendera.
\newblock $\mathrm{FO}^2$ with one transitive relation is decidable.
\newblock In {\em Proc. of STACS}, pages 317--328, 2013.

\end{thebibliography}

\begin{comment}
\newpage 

\section{Appendix}

\vspace{1.7cm}

\noindent
\textbf{Proof of Proposition \ref{firstproof}}
\begin{proof}
We restrict attention 
to vocabularies with at most 
binary relations and show that $\mathrm{FO}^2$
and $\mathrm{FU}_1$ are equiexpressive.
It is easy to observe that $\mathrm{FO}^2$ is in fact a
syntactic fragment of $\mathrm{FU}_1$.
Thus we only need to consider
expressing $\mathrm{FU}_1$-properties in $\mathrm{FO}^2$.
The nontrivial part in translating $\mathrm{FU}_1$
into $\mathrm{FO}^2$ involves dealing with the formulae of
the type $\exists x_1...\exists x_k\, \varphi$,
where $\varphi$ is a Boolean combination of formulae of
the following kind.
\begin{enumerate}
\item
Binary atoms of the type $R(x_i,x_j)$ and of the type $x_i=x_j$.
Each of these atoms binds
\emph{exactly} the same two variable symbols. (This is
due to  the uniformity condition of $\mathrm{FU}_1$.)
\item
Formulae $\chi(x_i)$ with at most one free variable.
\end{enumerate}
We next \emph{sketch} how to translate such a formula
$\exists x_1...\exists x_k\ \varphi$ into $\mathrm{FO}^2$.
First put $\varphi$ into disjunctive normal form and then
distribute the existential quantifier block over the disjunctions.
Therefore
$$\exists x_1...\exists x_k\varphi
\equiv (\exists x_1...\exists x_k\psi_1)\vee...
\vee(\exists x_1...\exists x_k\psi_n),$$
where each $\psi_i$ is now a \emph{conjunction} of
formulae of the following kind.
\begin{enumerate}
\item
Binary atoms of the type $R(x_i,x_j)$ and of the type $x_i=x_j$,
and also negations of such atoms. Each of these atoms binds
\emph{exactly} the same two variable symbols,
due to  the uniformity condition of $\mathrm{FU}_1$.
\item
Formulae $\chi(x_i)$ with at most one free variable.
\end{enumerate}
It is now easy to see how to translate such a
disjunct $\exists x_1...\exists x_k\psi_i$ into $\mathrm{FO}^2$.
For example, if
$\exists x_1...\exists x_k\psi_i$ is the
formula
$$\exists x_1 \exists x_2\bigl( R(x_0,x_1)\wedge \neg S(x_0,x_1)
\wedge \chi_0(x_0)\wedge \chi_1(x_1)\wedge \chi_2(x_2)\,\bigr),$$
it translates first into  
$$\exists x_1\bigl( R(x_0,x_1)\wedge \neg S(x_0,x_1)
\wedge \chi_0(x_0)\wedge \chi_1(x_1)\,\bigr)
\wedge \exists x_2\chi_2(x_2),$$
and since we know (by the induction hypothesis) how  to
translate the formulae $\chi_i(x_i)$ into
equivalent $\mathrm{FO}^2$-formulae $\chi_i^*(x_i)$,
we thereby obtain the $\mathrm{FO}^2$-formula
$$\exists x_1\bigl( R(x_0,x_1)\wedge \neg S(x_0,x_1)
\wedge \chi_0^*(x_0)\wedge \chi_1^*(x_1)\,\bigr)
\wedge \exists x_0\chi_2^*(x_0).$$
Therefore it is easy to observe 
that $\mathrm{FU}_1$ indeed translates into $\mathrm{FO}^2$.\qed
\end{proof}
%


%
\noindent
\textbf{Proof of Theorem} \ref{timegoesbytheorem}
\begin{proof}
The proof of this theorem is, up to an extent, similar
to that of the proof of Theorem \ref{firstproof} above.
The claim of the current theorem extends the 
observation from \cite{twovariablemodal} that
Boolean modal logic with the inverse operator and
with the identity modality, is expressively complete
for $\mathrm{FO}^2$. In the higher arity context of
the current theorem, the surjections in $\mathrm{SRJ}$
play a role analogous to the inverse operator
for binary relations.
To establish the claim of Theorem \ref{timegoesbytheorem},
we need to show that for each 
concept of $\mathcal{DL}_{\mathrm{FU}_1}$ there exists an
equivalent $\mathrm{FU}_1$-formula $\varphi(x)$ with
one free variable $x$, and for each $\mathrm{FU}_1$-formula $\psi(x)$,
there exists an equivalent concept of $\mathcal{DL}_{\mathrm{FU}_1}$.
The translation from $\mathcal{DL}_{\mathrm{FU}_1}$
into $\mathrm{FU}_1$ is
an obvious extension of the standard
translation of modal logic, so we only
discuss the translation from $\mathrm{FU}_1$
into $\mathcal{DL}_{\mathrm{FU}_1}$.
Let $X := \{x_1,...,x_m\}$ be a
set of $m\geq 2$ distinct variables.
We next define a translation
$r[x_1,...,x_m]$ that translates
Boolean combinations of $X$-atoms of $\mathrm{FU}_1$
into $m$-ary $\mathcal{DL}_{\mathrm{FU}_1}$ roles.
\begin{enumerate}
\item
Let $R(x_{i_1},...,x_{i_k})$ be an $X$-atom.
Then $(R(x_{i_1},...,x_{i_k}))^{r[x_1,...,x_m]} = \sigma R$,
where $\sigma:[k]\rightarrow[m]$ is
the surjection such that $\sigma(j) = i_j$.
For example, $(Q(y,x,x))^{r[x,y]} = \sigma Q$
for $\sigma = \{ (1,2),(2,1),(3,1) \}$.
\item
$(\mathcal{T}\wedge\mathcal{S})^{r\overline{x}}
:= \mathcal{T}^{r\overline{x}}\wedge\mathcal{S}^{r\overline{x}}$
and $(\neg \mathcal{T})^{r\overline{x}}
:= \neg (\mathcal{T}^{r\overline{x}})$,
where $\overline{x}$ denotes $[x_1,...,x_m]$.
\end{enumerate} 
We also define $(x=y)^{r[x,y]} :=\varepsilon$
and $(y=x)^{r[x,y]} := \varepsilon$.
With the translation $r$ defined,
we next define a translation $\cdot^*$
that translates formulae of $\mathrm{FU}_1$
into equivalent $\mathcal{DL}_{\mathrm{FU}_1}$-concepts.
We note that due to the 
way the syntax of $\mathrm{FU}_1$ was
defined, $\mathrm{FU}_1$-formulae
can have subformulae that are \emph{not}
themselves $\mathrm{FU}_1$-formulae.
The operator $\cdot^*$ will
be defined inductively only for
each formula constrution rule of $\mathrm{FU}_1$.
Thus we obtain rules for atoms, Boolean connectives and
existential block quantification.
We let $\top$ and $\bot$ translate to themselves and
define $(P(x))^*:= P$, where $P$ on the right hand side is now an atomic concept.
For unary atoms $R(x,...,x)$ with the free variable $x$, we define $(R(x,...,x))^*
:= \exists(\varepsilon\cap \sigma(R)).\top$,
where $\sigma$ is any surjection from $[\mathit{arity}(R)]$ onto $\{1,2\}$.
We let $(\neg\varphi(x))^* := \neg((\varphi(x))^*)$
and $(\varphi(x)\wedge\psi(x))^* = (\varphi(x))^*\wedge(\psi(x))^*$.
We note here that for example
the formula $P(x)\wedge Q(y)$ with two free variables is a
legitimate $\mathrm{FU}_1$-formula, but it is easy to check that
we do not need a translation rule for such formulae because
we are only interested in $\mathrm{FU}_1$-formulae
with at most one free variable.
This holds despite the formula $P(x)\wedge Q(y)$ and similar formulae
can appear as a subformulae in $\mathrm{FU}_1$-formulae with at
most one free variable;
consider for example the formula $\exists y( R(x,y)\wedge P(x)\wedge Q(y))$.
We will next \emph{sketch} 
how to translate an $\mathrm{FU}_1$-formula $\psi(x_0) :=
\exists x_1...\exists x_k\, \varphi$ with the free variable $x_0$.
We first put the subformula $\varphi$
into disjunctive normal form $\psi_1\vee...\vee\psi_p$,
where each $\psi_i$ is a conjunction.
Note that due to  the syntax of $\mathrm{FU}_1$, each of
the higher arity atoms in $\psi_1\vee...\vee\psi_p$ has
exactly the same set of first-order variables.
We have $\varphi(x_0) \equiv \exists x_1...\exists x_k(
\psi_1\vee...\vee\psi_p)$,
and thus $\varphi(x_0) \equiv (\exists x_1...\exists x_k\psi_1)
\vee...\vee(\exists x_1...\exists x_k\psi_p)$.
Consider a single disjunct
$\varphi_i(x_0) := \exists x_1...\exists x_k\psi_i$.
We have
$$\varphi_i(x_0)\, :=\, \exists x_1...\exists x_k
\bigl( \mathcal{T}\{y_1,...,y_n\}\wedge
\chi_1(u_1)\wedge...\wedge\chi_m(u_m)\bigr),$$
where both sets $\{y_1,...,y_n\}$ and $\{u_1,...,u_m\}$
are subsets of $\{x_0,...,x_k\}$,
and $\mathcal{T}\{y_1,...,y_n\}$ is a 
conjunction of higher arity atoms each with exactly the
set $\{y_1,...,y_k\}$ of distinct variables.
$\mathcal{T}\{y_1,...,y_n\}$ could also  be 
the formula $\top$, but we can ignore this case
without loss of generality.
Each $\chi_j(u_j)$
is an $\mathrm{FU}_1$-formula with at most one free variable.
We assume, w.l.o.g.,
that $\{y_1,...,y_n\}\subseteq\{u_1,...,u_m\}$.
We then consider the case where $x_0\in\{y_1,...,y_n\}$.
By symmetry, we may assume that $x_0 = y_1$.
We therefore have
$$\varphi_{i}(x_0)\, \equiv
\exists y_2...\exists y_n\bigl(\mathcal{T}\{y_1,...,y_n\}
\wedge{\chi}_1(y_1)\wedge...\wedge\chi_n(y_n)\bigr)
\ \wedge\ \Gamma,$$
where $\Gamma$ is a conjunction of
formulae of type $\exists v\chi_j(v)$.
The formulae $\exists v\chi_j(v)$ translate
to $\exists U.((\chi_j(v))^*)$,
where $U$ is the \emph{universal role} $U := \varepsilon\cup\neg\varepsilon$
and $(\chi_j(v))^*$ is the translation of $\chi_j(v)$.
We translate the formula 
$$\exists y_2...\exists y_n\bigl(\mathcal{T}\{y_1,...,y_n\}
\wedge{\chi}_1(y_1)\wedge...\wedge\chi_n(y_n)\bigr)$$
to $C_1\sqcap\exists\bigl(\, \mathcal{T}
\{y_1,...,y_n\}^{r[y_1,...,y_n]}\, \bigr)(C_2,...,C_n)$,
where $\cdot^r$ is the relation translation operator we defined above
and where $C_i = (\chi_i(y_i))^*$ for each $i$.
Consider then the case where $x_0\not\in\{y_1,...,y_n\}$.
This time
$$\varphi_{i}(x_0)\, \equiv
\exists y_1...\exists y_n\bigl(\mathcal{T}\{y_1,...,y_n\}
\wedge{\chi}_1(y_1)\wedge...\wedge\chi_n(y_n)\bigr)
\ \wedge\ \Gamma',$$
where $\Gamma'$ is a conjunction of
formulae of the type  $\exists v(\chi_i(v))$
with an additional formula $\chi(x_0)$.
%
%
We translate $\exists v(\chi_i(v))$ to $\exists U.((\chi(v))^*)$ as
before, and $\chi(x_0)$ to $(\chi(x_0))^*$.
We still need to translate
the formula 
$$\beta := \exists y_1...\exists y_n\bigl(\mathcal{T}\{y_1,...,y_n\}
\wedge{\chi}_1(y_1)\wedge...\wedge\chi_n(y_n)\bigr),$$
which does not have free variables.
We remove the quantifier $\exists y_1$ from
the  prefix $\exists y_1...\exists y_n$
and translate the  resulting formula to
$$D\, :=\, C_1\sqcap\exists\bigl(\, \mathcal{T}
\{y_1,...,y_n\}^{r[y_1,...,y_n]}\, \bigr)(C_2,...,C_n),$$
where $C_i = (\chi_i(y_i))^*$ for each $i$.
We let $(\beta)^* := \exists U.D$.
We have now shown
how to translate the formula $\psi(x_0) := \exists x_1...\exists x_k\varphi$,
where $x_0$ is a free variable in $\psi(x_0)$.
To translate a formula $\exists x_0...\exists x_k\varphi$ without a
free variables, we apply arguments virtually identical to
the ones already presented above.\qed
\end{proof}
%
%


\end{document}